\documentclass{article}

\usepackage{fullpage}

\usepackage{amsmath,amssymb,amsfonts}
\usepackage{algorithmic}
\usepackage{graphicx}
\usepackage{textcomp}
\usepackage{xcolor}
\usepackage{tikz}
\usepackage{hyperref}
    
\usepackage{amsthm}

\newtheorem{lemma}{Lemma} 
\newtheorem{proposition}{Proposition} 
\newtheorem{theorem}{Theorem} 
\newtheorem{corollary}{Corollary} 
\newtheorem{claim}{Claim}

\theoremstyle{definition}
\newtheorem{definition}{Definition}
\newtheorem{example}{Example}  
\newtheorem{remark}{Remark}    

\DeclareMathOperator\eqp{\mathrel{\stackrel{\mbox{\normalfont\tiny $\log$}}{=}}}

\DeclareMathOperator\gep{\mathrel{\stackrel{\mbox{\normalfont\tiny $\log$}}{\ge}}}
\DeclareMathOperator\lep{\mathrel{\stackrel{\mbox{\normalfont\tiny $\log$}}{\le}}}

\newcommand{\pr}{\mathrm{Prob}}

\def\firstcircle{(150:1.75cm) circle (2.5cm)}
\def\secondcircle{(30:1.75cm) circle (2.5cm)}
\def\thirdcircle{(270:1.75cm) circle (2.5cm)}

\begin{document}

\title{Common information in well-mixing graphs  and applications to information-theoretic cryptography\footnote{%
This document is an extended version of the conference paper published in the 
\emph{Proceedings of the IEEE Information Theory Workshop}, Shenzhen, 2024,
\href{https://doi.org/10.1109/ITW61385.2024.10806994}{DOI: 10.1109/ITW61385.2024.10806994}.
In this version, we provide complete proofs of all results along with more detailed discussions.
Furthermore, we establish a bound on the communication complexity of secret key agreement (see Section~5) in a more general setting, which also covers communication protocols with private randomness.}
}

\author{Geoffroy Caillat-Grenier
\and 
Andrei Romashchenko
\and
Rustam Zyavgarov%
}

\maketitle

\begin{abstract}
We study the connection between mixing properties for bipartite graphs and materialization of the mutual information in one-shot settings. We show that mixing properties of a graph imply impossibility to extract the mutual information shared by the ends of an edge randomly sampled in such a graph. We apply these impossibility results to some questions motivated by information-theoretic cryptography. In particular, we show that communication complexity of a secret key agreement in one-shot setting is inherently uneven: for some inputs, almost all communication complexity inevitably falls on only one party (even for protocols with private randomness). Our main results are formulated and proven in the language of Kolmogorov complexity;  we further develop tools to reinterpret some of them in terms of Shannon entropy.
\end{abstract}

%\begin{keyword}
%Kolmogorov complexity \sep algorithmic information theory \sep communication complexity \sep  secret key agreement \sep common information
%\end{keyword}

\section{Introduction}

In classical information theory, Shannon's mutual information between two random variables is both a fundamental concept and a powerful analytical tool. However, the definition of mutual information lacks a direct physical interpretation: the value of the mutual information between two sources of information does not necessarily correspond to any common data set shared by these sources. The question of a more direct and ``material'' interpretation of mutual information has been raised since the 1970s. In this context, Gács and Körner introduced the notion of common information of two random objects~\cite{gk1973}. The common information by Gács and Körner corresponds, loosely speaking, to the greatest common part that can be “easily” extracted from both correlated random sources.

Speaking slightly more precisely, the Gács--Körner common information of two sources $\cal X$ and $\cal Y$ is the maximal information contained in a $\cal Z$ such that $\cal Z$ contains negligibly small information conditional on $\cal X$ and conditional on $\cal Y$. To convert this phrase into a formal definition, we need to quantify the words “negligibly small,” and a neat and exact definition is inevitably technical.
The resulting quantity (the physically materializable part of the mutual information shared by two given sources) is below (and in some cases far below) the value of Shannon's mutual information. In a similar vein, Wyner studied the question of the smallest (simplest) object with respect to which a given pair of information sources would become independent~\cite{wyner}. Wyner's common information for two random sources is above (in some cases far above) Shannon's mutual information.

 \smallskip
 
Analogously, one of the central concepts in algorithmic information theory is Kolmogorov's mutual information.
Formally, the mutual information between bit strings $x$ and $y$ is defined as the difference between the total length of optimal programs that produce $x$ and $y$ separately, 
and the length of an optimal program that produces both $x$ and $y$ together.
This definition also lacks a concrete or intuitive interpretation.
In their seminal paper, G\'acs and K\"orner proposed a version of their notion of common information within the framework of Kolmogorov complexity.
They defined the common information of two binary strings $x$ and $y$ as the maximum complexity of a string $z$ that can be easily derived from either $x$ or $y$; that is, the conditional complexities of $z$ given $x$, and $z$ given $y$, are both small.
It is known that for certain pairs of strings, the value of this common information can be significantly smaller than the mutual information.
Wyner's notion of common information can also be adapted to the setting of Kolmogorov complexity:
given strings $x$ and $y$, one can minimize the complexity of a string $z$ such that $x$ and $y$ become conditionally independent given $z$; that is, the sum of the lengths of optimal programs transforming $z$ into $x$ and $y$ separately is essentially equal to the length of an optimal program transforming $z$ into the pair $(x,y)$.

The concept of common information is inherently subtle, and the balance between maximizing the complexity of a string $z$ and allowing some slack in its conditional complexities given $x$ and $y$ can involve intricate trade-offs,
see, e.g, \cite{muchnik-common-info,cmrsv,andrei-pairs}. 
In this paper, we investigate certain extremal cases---pairs of strings $(x,y)$ whose mutual information is largely non-materializable.
Loosely speaking, we show that for certain natural examples of pairs $(x,y)$  a string $z$ can absorb the entire mutual information between $x$ and $y$ only if it essentially contains either $x$ or $y$.

A distinctive feature of our approach is that it explores how the materialization of mutual information is linked to the combinatorial structure of the underlying objects.
We focus on pairs $(x,y)$ that can be interpreted as a ``typical'' edge of a graph with certain specific combinatorial properties.
In this context, revealing Wyner's common information corresponds to identifying an induced subgraph that is nearly a complete bipartite graph.

We establish several lower bounds on the size of such subgraphs and, correspondingly, impossibility results for the materialization of mutual information,
for graphs with strong mixing properties (see Theorems~\ref{th:1} and~\ref{th:2} below).
To capture the intuitive idea of mixing in a graph, we adopt two non-equivalent but complementary formalizations.
Our first argument (Theorem~\ref{th:1}) develops the method proposed in~\cite{muchnik-common-info,cmrsv}, utilizing bounds on the number of common neighbors shared by vertex pairs.
The other argument (Theorem~\ref{th:2}) relies on the spectral properties of the graph, which also capture its underlying mixing behavior.

Throughout the paper, we refer to the same recurrent example (or, more precisely, a class of examples) to which both of our approaches apply.
In this example, one of the two objects is a polynomial of bounded degree over a finite field, and the other is an evaluation point for this polynomial (a point on the graph of the polynomial).
Obviously, such objects share a large amount of mutual information. However, it is extremely hard to materialize it.
Both of our approaches work for this class of examples:
the bounded-neighborhood property follows from the fact that two low-degree polynomials cannot share many common evaluation points; the spectral bound for the underlying graph is less obvious, but it follows from a routine calculation. We believe that these examples (a polynomial and its evaluation points) are not only natural and intrinsically interesting, but might also have applications in adjacent areas, such as coding theory and cryptography.

Subsequently, we compare the methods underlying Theorems~\ref{th:1} and~\ref{th:2}. We employ both of them to derive information-theoretic inequalities tailored to graphs with mixing behavior (see Theorem~\ref{th:rustam}).
Our analysis shows that these two approaches capture different aspects of mixing in a graph, and that their applications are incomparable. We apply the established bounds to certain questions in information-theoretic cryptography, which we briefly outline below.

\smallskip
 
\emph{Application I.} 
Our first application concerns the problem of secret key agreement. In the setting introduced in~\cite{secret-key-1,secret-key-2}, one party (Alice) receives a random input $x$, while the other party (Bob) receives a correlated random input $y$. Their goal is to agree on a shared secret key $z$ through communication over a public channel. It is known that Alice and Bob can agree on a secret key whose length is approximately equal to the mutual information between $x$ and $y$. The classical communication protocol that achieves this goal (see the original construction in~\cite{secret-key-1,secret-key-2} and its adaptation\footnote{%
Usually, information-theoretic cryptography is built on the foundations of Shannon's theory. However, an alternative approach based on Kolmogorov complexity has its advantages.  
This approach was originally suggested in~\cite{antunes}.  
See also~\cite{andrei-marius,emirhan} for adaptations of this approach specifically to the problem of secret key agreement. %
} 
to the framework of Kolmogorov complexity in~\cite{andrei-marius}) can be described as follows. One of the parties (say, Alice) sends a hash of her input to the other party, with the hash size approximately equal to the conditional complexity of $x$ given $y$. This allows Bob to reconstruct $x$; this phase of the protocol is called information reconciliation.
Subsequently, both parties extract from $x$ a shared random key of length roughly equal to the mutual information between $x$ and $y$, such that this key is nearly independent of the information leaked to an eavesdropper; this phase is referred to as privacy amplification.

We prove that for certain classes of inputs $(x,y)$, this standard (highly asymmetric) protocol is essentially the only possible solution. We show (Theorem~\ref{th:secret_key_agreement}) that for certain pairs of inputs, it is necessary for either Alice to send Bob all the information required to reconstruct $x$
(conditional complexity of $x$ given $y$), 
or for Bob to send Alice all the information required to reconstruct $y$ (conditional complexity of $y$ given $x$); 
there is no way to divide the communication load more evenly between the two parties. Our result holds even when the protocol allows private randomness, that is, when Alice and Bob, in addition to their inputs $x$ and $y$, each have access to private sources of random bits, and these bits may be utilized during the protocol.

While the proof presented in this paper applies to incidences (lines and points) in a plane over finite fields of arbitrary characteristic, our argument is restricted to protocols that achieve an optimal key size, i.e., one that is close to the mutual information between Alice's and Bob's inputs.
It was recently shown that for protocols agreeing on keys of suboptimal size, a similar bound holds only for lines and points over prime fields; see \cite{mfcs2025}.

\smallskip
 
\emph{Application II.} Our second application deals with the setting proposed by An.~Muchnik in \cite{muchnik-crypto}. We assume that Sender wants to transmit to Receiver a \emph{clear message} $x$ in an encrypted form. The encryption is based on the assumption that Sender and Receiver share a \emph{secret key} $y$. The challenge is to prepare an encrypted message that reveals no information on $x$ to Eavesdropper who owns side information $z$ correlated with $(x,y)$. We show that for some correlation between $(x,y)$ and $z$ it is impossible to find an encrypted message that allows Receiver to compute the clear message $x$ and, at the same time, reveals no information on $x$ to Eavesdropper
(see Proposition~\ref{p:muchnik} and Proposition~\ref{prop:4}). 
It is crucial that this effect is explained by the combinatorial nature of the correlation between $x$, $y$, and $z$, and not simply by the entropy values of these quantities.
This result allow us to answer an open question from \cite{vereshchagin_crypto}.
 
 \smallskip

\emph{From Kolmogorov to Shannon.}
Our main results are formulated and proven in the framework of Kolmogorov complexity. However, some of our results can be transposed into the framework of classical Shannon information theory. We show (see Section~\ref{s:shannon}) that our \emph{Application I} (an impossibility result for secret key agreement) implies a homologous statement in the probabilistic framework.  
In this setting, Alice and Bob receive as inputs values of correlated random variables $\mathcal{X}$ and $ \mathcal{Y}$, respectively; for example, one of them receives a random polynomial, and the other a random evaluation point of this polynomial.  
The goal is to use communication over a public channel to produce a random variable $\cal Z $ with high entropy and negligible mutual information with the transcript of the protocol (which is revealed to the eavesdropper).  
We prove that for some probability distributions on the inputs (including our recurrent example---a random polynomial and its random evaluation point), such secret key agreement is only possible if either Alice sends approximately $ H(\mathcal{X} | \mathcal{Y}) $ bits of information to Bob, or Bob sends approximately $ H(\mathcal{Y}|  \mathcal{X}) $ bits of information to Alice.  
Our proof is based on a fairly standard connection between Kolmogorov complexity and Shannon entropy.

Although we show that, for our setting, in any communication between Alice and Bob either Alice sends approximately $ H(\mathcal{X} | \mathcal{Y}) $ bits, or Bob sends approximately $ H(\mathcal{Y} | \mathcal{X}) $ bits,  
we cannot claim that the average size (and the entropy) of Alice's and Bob's messages is as large as $ H(\mathcal{X} | \mathcal{Y})$ and $ H(\mathcal{Y} | \mathcal{X}) $, respectively.  
Indeed, there exist protocols such that in half of all cases Alice sends approximately $ H(\mathcal{X} | \mathcal{Y}) $ bits of information (and Bob sends nothing),  
and in the other half of the cases Bob sends approximately $ H(\mathcal{Y} | \mathcal{X}) $ bits of information (and Alice sends nothing).  
It seems that such a result (a ``non-convex'' lower bound for the sizes of Alice's and Bob's messages) cannot be proven directly using conventional inequalities for Shannon entropy (which are inherently convex).  
While the usual entropic tools do not suit this type of phenomenon well, Kolmogorov complexity is perfectly tailored for such a setting.  
Let us also mention that such ``non-convex'' behavior is specific to the one-shot setting and cannot occur when the inputs $ \mathcal{X} $ and $ \mathcal{Y} $ are produced by i.i.d.\ sources,  
where we can apply the standard time-sharing argument.

 \smallskip
 
The rest of the paper is organized as follows.  
In Section~\ref{s:preliminaries}, we introduce the main combinatorial, algebraic, and information-theoretic tools used in the proofs of our main results.  
In Section~\ref{s:mixing}, we prove our main results on non-extractability of mutual information (Theorems~\ref{th:1} and~\ref{th:2}).  
In Section~\ref{s:inequalities}, we transform the bounds from Section~\ref{s:mixing} into linear inequalities for Kolmogorov complexity and show that the techniques from Theorems~\ref{th:1} and~\ref{th:2} are incomparable,  
i.e., neither of the two methods is strictly stronger than the other.  
In Section~\ref{s:app1}, we use the bounds from Section~\ref{s:mixing} to prove negative results on the communication complexity of secret key agreement.  
In Section~\ref{s:shannon}, we use the result proven in Section~\ref{s:mixing} to establish a similar statement in the framework of Shannon's information theory.  
Finally, in Section~\ref{s:app2}, we prove impossibility results on conditional descriptions secure against an eavesdropper with pre-known information (in Muchnik's setting).

\subsection{Notation}

\begin{itemize}
\item $|\cal S|$ stands for the cardinality of a finite set $\cal S$
\item $G=(L, R; E)$ stands for a bipartite graph where $L \cup R$ (disjoint union) is  the set of vertices and $E\subset L \times R$ is the set of edges; for every vertex $v$ in the graph, $\Gamma(v)$  denotes the set of all neighbors of $v$
\item $H({\cal X})$ denotes Shannon's entropy of a random variable $\cal X$
\item $C(x)$ and $C(x|y)$ stand for Kolmogorov complexity of a string $x$ and, respectively, conditional Kolmogorov complexity of $x$ conditional on $y$, see \cite{li-vitanyi,suv}. We use a similar notation for more involved expressions, e.g., $C(x,y|v,w)$ stands for Kolmogorov complexity of the concatenation of $x$ and $y$ (or a code for the pair $(x,y)$) conditional on the concatenation (the pair) of $v$ and $w$ 
\item $I(x:y) := C(y)+C(y|x)$ and $I(x:y|z) := C(y|z) - C(y|x,z)$ stand for  information in $x$ on $y$ and, respectively, information in $x$ on $y$ conditional on $z$ \item $I(x:y:z) := I(x:y) - I(x:y|z)$ denotes the triple mutual information
\end{itemize}
Many natural equalities and inequalities for Kolmogorov complexity are valid only up to a logarithmic additive term, e.g.,
\[
C(x,y) = C(x) + C(y|x) \pm O(\log n),
\]
where $n$ is the sum of lengths of $x$ and $y$ (this is the chain rule a.k.a. Kolmogorov--Levin theorem, see \cite{zl}). To simplify the notation, we write $A\lep B$ instead of 
\[
A \le B + O(\log N),
\]
where $N$ is the sum of lengths of all strings involved in the expressions $A$ and $B$. 
Similarly we define $A\gep B$ (which means $B\lep A$) and $A\eqp B$ (which means $A\lep B$ and $B\lep A$). 
For example, the chain rule can be expressed as
\[
C(x,y) \eqp C(x)+C(y|x);
\]
the well known fact of symmetry of the mutual information can be expressed as
\[
I(x:y) \eqp C(x)+ C(y) - C(x,y),
\]
and a symmetric expression for the triple mutual information
\[
\begin{array}{rcl}
I(x:y:z) &\eqp& C(x)+C(y)+C(z)
-C(x,y) - C(x,z) - C(y,z) + C(x,y,z)
\end{array}
\]

\section{Preliminaries}
\label{s:preliminaries}

Here we define the mathematical objects that will be used in this paper and prove useful information-theoretic properties about them.

\begin{definition}
We say that a bipartite graph $G=(L,R;E)$ has $d$-\emph{bounded common neighborhood on the right} if for every pair of vertices $v_1,v_2\in R$, $v_1 \not= v_2$, there are at most $d$ vertices in $L$ that are neighbors of both $v_1$ and $v_2$.
\end{definition}
\begin{example}\label{e:poly}
Let $\mathbb{F}$ be a finite field, $L$ be the set of pairs $(x_1,x_2)\in \mathbb{F}^2$, and $R$ be the set of all polynomials
\[
S(t)= s_0 + s_1t + \ldots + s_d t^d
\]
 of degree at most $d$ over $\mathbb{F}$. Let $E$ consist of all edges that connect  a point $(x_1,x_2)$ from $L$ with a polynomial with coefficients $(s_0,s_1,\ldots, s_d)$ from $R$  if and only if 
\[
x_2 = s_0 + s_1x_1 + \ldots + s_d x_1^d, 
\]
i.e., $(x_1,x_2)$ belongs to the graph of the polynomial $ s_0 + s_1t + \ldots + s_d t^d$. This graph $G = (L,R;E)$ has $d$-bounded common neighborhood on the right since two polynomials of degree $d$ (representing two vertices in $R$) have at most $d$ points of intersection (common neighbors in $L$). 
\end{example}

Let $G=(L,R;E)$ be a bi-regular bipartite graph where all vertices in $L$ have degree $D_L$ and all vertices in $R$ have degree $D_R$. The adjacency matrix of such a graph is a symmetric square matrix of size $m\times m$, where $m = |L|+|R|$ is the total number of vertices. It is well known that all eigenvalues $\lambda_i$ of such a matrix are real numbers such that 
\[
\sqrt{D_L D_R}  = \lambda_1 \ge \lambda_2 \ge \ldots \ge \lambda_m = -\sqrt{D_L D_R},
\]
and $\lambda_i = -\lambda_{m-i+1}$ (the spectrum of a bipartite graph is symmetric), see, e.g. \cite{spectral-graph-theory}. The gap between $\lambda_1$ and $\lambda_2$ is called \emph{the spectral gap} of the graph. This value has deep combinatorial implications (smaller $\lambda_2$ implies better mixing properties of the graph, see \cite{expanders-survey}).

\begin{definition}
Let $G=(L, R; E)$ be a bipartite graph with degrees $D_L$ and $D_R$ for vertices in $L$ and $R$ respectively. We say that $G$ is a \emph{good spectral bipartite expander} if for the 2nd eigenvalue of this graph we have
$\lambda_2  \le \sqrt{\max\{ D_L, D_R\}}$.
\end{definition}
\begin{remark}
In contrast to a most conventional definition of an expander graph, we do not assume that $D_L$ and $D_R$ are bounded by $O(1)$. %; degree of the graph can be, e.g., square root of the number of vertices. 
\end{remark}

\begin{proposition}
\label{p:standard-profile}
Denote $n:=\log q$. Let $G = (L,R;E)$ be a graph 
such that
\[
|L| = q^2 , \ |R| = q^{d+1},
\]
and the degree of each  vertex in $R$ is equal to $q$, and degree of each vertex in $L$ is equal to $q^{d}$ (hence $|E| = q^{d+2}$).

Now, if we impose the uniform distribution on the set of edges $E$ of this graph, we obtain jointly distributed random variables  $({\cal X},{\cal Y})$ with Shannon's entropies 
\[
H({\cal X}) = 2n,\ H({\cal Y}) = (d+1)n,\ H({\cal X},{\cal Y}) = (d+2)n.
\]
Moreover, if the graph is given explicitly (the complete list of  vertices and edges of the graph can be found algorithmically given the parameters $n$ and $d$), 
then for the vast majority of pairs $(x, y) \in E$ we have
\begin{equation}
\label{eq:profile}
 C(x) \eqp 2n,\ C(y) \eqp (d+1)n,\ C(x,y) \eqp (d+2)n.
\end{equation}
\end{proposition}
\begin{proof}
The claim on the entropy values is straightforward. 
The equalities for Kolmogorov complexity follow from a standard counting argument (see e.g. \cite[Theorem 5]{suv}, and the discussion below).
\end{proof}
\begin{remark}
Proposition~\ref{p:standard-profile} applies to the graph from Example~\ref{e:poly} constructed for the field $\mathbb{F}$ of cardinality $q$.
\end{remark}
\begin{definition}
For a graph as in Proposition~\ref{p:standard-profile} we will say that an edge $(x,y)\in E$ is \emph{typical} if it satisfies ~\eqref{eq:profile}.
\end{definition}

\begin{lemma}
\label{p:spectrum}
The graph from Example~\ref{e:poly} constructed on the field $\mathbb{F}_{n}$ with $|\mathbb{F}_n| = q_n$ is a good spectral bipartite expander. In particular, this holds for $q_n = 2^n$.
\end{lemma}

\noindent See Appendix for the proof of Lemma~\ref{p:spectrum}.

\begin{definition}\label{d:privateRandomGraph}
Let $G = (L, R; E)$ be a bipartite graph, and let $m$ be an integer. 
We call by \emph{$m$-private randomness amplified $G$} 
the graph $\bar G = (\bar L, \bar R; \bar E)$, where 
\[
\begin{array}{rcl}
\bar L &:=& L\times\{0,1\}^m,\\
\bar R &:=& R\times\{0,1\}^m, 

\end{array}
\]
and the set of edges $\bar E$ consists of all the pairs $(\bar x, \bar y) = (\langle x, r_x \rangle, \langle y, r_y \rangle)$ such that $(x,y)$ is an edge in the original graph $G$ and $(r_x, r_y)$ is a pair of binary strings of size $m$. 
A matrix representation of $\bar G$ is the tensor product between $G$'s adjacency matrix and the all one matrix of size $2^m \times 2^m$.  
\end{definition}

This definition is useful to analyze communication protocols with private randomness. In a typical edge $(\langle x, r_x \rangle, \langle y, r_y \rangle)$ of $\bar G$, $x$ and $y$ can be interpreted as (correlated) inputs of two communicating parties, and $r_x$, $r_y$ can be understood  as outcomes of private sources of random bits used by the parties.

In this paper, we will always consider implicitly that the number of random bits $m$ is polynomial in the size of the graph considered (this is necessary for Lemma~\ref{l:induced-edges} to hold). 
However, this restriction can be waived. Indeed, a version of Newman's theorem 
(see \cite{newman} for the standard version; in our setting we need the adaptation proposed in\cite[Proposition 1]{emirhan}) claims that, roughly speaking, for any protocol of secret key agreement using private random bits, there is another protocol that solves the same communication problem and uses only a polynomial number of random bits. 

\begin{proposition}
\label{p:standard-aug-profile}
Let $G = (L,R;E)$ be a graph as in Proposition~\ref{p:standard-profile}, and
 let  $\bar G = (\bar L, \bar R; \bar E)$ be an $m$-private randomness amplified $G$.
If $G$ is given explicitly (the complete list of  vertices and edges of the graph can be found algorithmically given the parameters $n$ and $d$), 
then for the vast majority of pairs $(\bar x, \bar y) \in \bar E$ we have
\begin{equation}
\label{eq:profile-aug}
 C(\bar x) \eqp 2n+m,\ C(y) \eqp (d+1)n+m,\ C(x,y) \eqp (d+2)n+2m.
\end{equation}
\end{proposition}
\begin{proof}
We can determine the number of vertices in $\bar L, \bar R$ and the pairs in $\bar E$ and then conclude by using the standard counting argument from \cite[Theorem 5]{suv}. Every vertex in the $G$ has $2^m$ copies of itself in $\bar G$, hence $|\bar L| = 2^{2n + m}$ and $|\bar R| = 2^{(d+1)n + m}$. For every edge in $E$, there are $2^{2m}$ corresponding edges in $\bar E$, hence $|\bar E| = 2^{(d+2)n + 2m}$.
\end{proof}

\noindent We get the analogous definition of typicality for vertices in $\bar G$ by using \eqref{eq:profile-aug} instead of \eqref{eq:profile}.

Now we introduce some notation that will be crucial in the proof of our main technical results (Theorem~\ref{th:1} and Theorem~\ref{th:2}).

\begin{definition}\label{d:partitionsGraph}
Let  $ G=(L, R; E)$ be the graph from Example~\ref{e:poly} and let  $\bar G = (\bar L, \bar R; \bar E)$ be the $m$-private randomness amplified $G$.
We fix a pair  $(\bar x, \bar y) \in \bar E$ satisfying \eqref{eq:profile-aug} and an arbitrary bit string $w$.

(a) Now we define a subgraph in $\bar G$ that consists of vertices that are \emph{mildly similar to $(\bar x, \bar y)$ when $w$ is known}. Formally, we let $\bar G_w=(\bar L_w, \bar R_w; \bar E_w)$ be the induced subgraph in $\bar G$ that contains the vertices 
\begin{equation}\label{partitionsGraph}
\begin{array}{rcl}
\bar L_w &=& \{\bar  x'\in \bar L\ :\ C(\bar x'|w) \lep C(\bar x|w) \}, \\
\bar R_w &=& \{\bar  y'\in \bar R\ :\ C(\bar y'|w) \lep C(\bar y|w) \},
\end{array}
\end{equation}
and the edges 
\[
\bar E_w = E \cap (\bar L_w \times \bar R_w).
\]

(b) In a similar way, we define a subgraph in $\bar G$ that consists of edges that are \emph{strongly similar to $(\bar x, \bar y)$ when $w$ is known}.
To this end, we let $\hat G_w=(\bar L_w, \bar R_w; \hat E_w)$  that consists of the same sets of vertices  the vertices $\bar L_w \subset \bar L$ and  $\bar R_w\subset \bar R$, and a more restricted set of edges $\hat E_w$ defined as 
\[
\hat E_w = \{ (\bar x', \bar y') \in  \bar E_w \text{ such that } C(\bar x' |\bar y', w) \le C(\bar x|\bar y, w) \text{ and }  C(\bar y'|\bar x', w) \le C(\bar y|\bar x, w)\}.
\]
\end{definition}
While discussing the complexity profile of $\bar x, \bar y, w$, we will use the following notation:
\[
\begin{array}{l}
\alpha:= I(w:\bar x|\bar y),\ \beta := I(w:\bar y|\bar x), \ \text{and}\  \gamma:= I(\bar x:\bar y:w),
\end{array}
\]
see the Venn diagram in Fig.~\ref{figure1}. Observe that
\[
\begin{array}{rcl}
C(\bar x|w) &\eqp& C(\bar x|w,\bar y) + I(\bar x:\bar y|w)  = 2n+m -\alpha-\beta, \\
C(\bar y|w) &\eqp& C(\bar y|w,\bar x) + I(\bar x:\bar y|w)  = (d+1)n+m -\beta - \gamma,\\
C(\bar x, \bar y|w) &\eqp& C(\bar x|w,\bar y) + C(\bar y|w,\bar x) + I(\bar x:\bar y|w)  = (d+2)n+2m -\alpha-\beta - \gamma.
\end{array}
\]
\begin{figure}[htbp]
\centering
				\begin{tikzpicture}[scale=0.75]
				  \draw \firstcircle node[above left] {\small $m + n-\alpha$};
				  \draw \secondcircle node [above right] {\small $m + dn-\gamma$};
				  \draw \thirdcircle node [below] {\small \ldots};
				  \node at (95:0.05)   {\small $\beta$};
				  \node at (90:1.55) {\small $n-\beta$};
				  \node at (210:1.75) {\small $\alpha$};
				  \node at (330:1.75) {\small $\gamma$};
                  
				  \node at (150:4.75) {\Large $\bar x$};
				  \node at (30:4.75) {\Large $\bar y$};
				  \node at (270:4.75) {\Large $w$};
				\end{tikzpicture}
\caption{
Complexity profile of $(\bar x,\bar y,w)$ such that 
\(C(\bar x|\bar y) = m + n,\ C(\bar y|\bar x) = m + dn,\ I(\bar x:\bar y) = n,\)
and 
\(I(\bar x:w|\bar y) = \alpha,\ I(\bar x:\bar y:w)= \beta,\ \text{and}\ I(\bar y:w|\bar x) =\gamma\).}
\label{figure1}
\end{figure}

\begin{remark}
In the diagram in Fig.~\ref{figure1} we do not specify the value $C(w|\bar x,\bar y)$, which cannot be found given the parameters $n,d,\alpha,\beta$. As the graph $\bar G$ is given explicitly (i.e., if we can find it algorithmically given the parameters), applying the technique of \cite[theorem~5]{muchnik-romashchenko}, one can show that for every $w$ there exists a $w'$ such that all information quantities shown in Fig.~\ref{figure1}  for $(\bar x,\bar y,w')$
remain the same (up to $O(\log N)$) for $(\bar x,\bar y,w)$ except for the undefined term $C(w|\bar x,\bar y)$, which  vanishes, i.e., $C(w'|\bar x,\bar y) \eqp 0$. However, we do not use this observation in our arguments.
\end{remark}

We have the two following lemmas:

\begin{lemma}\label{l:standard-upper-bound}
For the sets $\bar L_w$ and $\bar R_w$ from Definition~\ref{d:partitionsGraph} 
we have   
\[
|\bar L_w| \le 2^{C(\bar x|w)+1} \ \text{and}\ |\bar R_w| \le 2^{C(\bar y|w)+1}.
\]
\end{lemma}
\begin{proof}

In general, for every string $b$ and for every integer number $\ell$, 
the number of all strings $a$ such that $C(a|b) \le \ell$ is not greater than the number of binary descriptions $w$ of length  at most $\ell$, which is
\[
1+2+\ldots+2^\ell = 2^{\ell+1}.
\]
We apply this observation to count the strings $x'$ and $y'$ provided the upper bounds on  $C(x'|w)$ and $C(y'|w)$, and we are done.
\end{proof}
\begin{lemma}\label{l:induced-edges}
For the sets $\bar L_w$, $\bar R_w$, $\bar E_w$, and $\hat E_w$ from Definition~\ref{d:partitionsGraph}  we have
\[
|\bar L_w| \ge 2^{C(\bar x|w) - O(\log N)},\ |\bar R_w| \ge 2^{C(\bar y|w) - O(\log N)}, \text{ and } |\bar E_w| \ge |\hat E_w| \ge 2^{C(\bar x, \bar y|w) - O(\log N)}.
\]
\end{lemma}
\begin{proof}
This argument is also pretty standard but slightly subtler than the observation made above on the \emph{upper} bounds for the cardinalities of $\bar L_w$ and $\bar R_w$. 
Given the graph $G=(L, R; E)$ and the binary expansion of $m$, 
we can construct the graph $\bar G = (\bar L,\bar R ; \bar E)$.  (Observe that the total description of these parameters requires $O(\log N)$ bits of information). 

If we are given also a string $w$, and the parameters $C(\bar x|w)$ and $C(\bar y|w)$, we can take all programs of length  $C(\bar x|w)$ and of length  $C(\bar y|w)$, and run them in parallel on input $w$. As some of these programs stop, we reveal one by one the elements of $\bar L_w$ and $\bar R_w$ and, therefore, edges in $\bar E_w = (\bar L_w \times \bar R_w) \cap \bar E$. The original vertices $\bar x$, $\bar y$, and the edge  $(\bar x,\bar y)$ eventually appear in this enumeration process.

To enumerate $\hat E_w$, we need on top of everything to verify for each revealed edge $(\bar x', \bar y')\in \bar E_w$ whether 
\[
C(\bar x'|\bar y', w) \le C(\bar x|\bar y, w) \text{ and }  C(\bar y'|\bar x', w) \le C(\bar y|\bar x, w).
\]
This property is enumerable (i.e., we can reveal all \emph{positive} answers) if we know the binary expansions of the numbers $C(\bar x|\bar y, w)$ and $ C(\bar y|\bar x, w)$: we run all the programs of size smaller than $C(\bar x| \bar y, w)$ (respectively $ C(\bar y|\bar x, w)$) with input $(\bar y', w)$ (respectively $(\bar x', w)$) until one of them returns $\bar x'$ (respectively $\bar y'$).

To identify any specific vertex $\bar x'$ in $\bar L_w$ given $w$, we need to know the parameters required to launch the enumeration described above and the position (the index) of this vertex in our enumeration of $\bar L_w$. Such an index would consist of $\log |\bar L_w|$ bits. Since we know that $\bar x \in \bar L_w$, we have
\[
C(\bar x|w) \lep \log |\bar L_w|.
\] 
Reading this inequality ``from the right to the left,'' we obtain 
$
|\bar L_w| \ge 2^{C(\bar x|w) - O(\log N).}
$
A similar argument gives 
$
|\bar R_w| \ge 2^{C(\bar y|w) - O(\log N).}
$

In a similar way, to identify any specific edge in $\hat E_w$ given $w$, we need to specify the parameters used to run the enumeration process (a description of the graph $G$ and the binary expansions of the values $C(\bar x|w)$, $C(\bar y|w)$, $C(\bar x|\bar y, w)$,   $C(\bar y|\bar x, w)$) and provide the index of this edge in the enumeration of $\hat E_w$. This index consists of $\log |\hat E_w|$ bits.
Therefore, as $(\bar x,\bar y)\in \hat E_w$, we have
\[
C(\bar x,\bar y|w) \lep \log |\hat E_w|.
\]
It follows that 
$
|\hat E_w| \ge 2^{C(\bar x, \bar y|w) - O(\log N)},
$
and we are done. Finally, the bound $|\bar E_w| \ge |\hat E_w|$ is obvious.
\end{proof}

\section{From mixing properties to information inequalities}
\label{s:mixing}

In what follows, we will prove two theorems that both are statements about \emph{mutual information inextractability}. Roughly speaking, we are going to show that the mutual information shared by a typical edge $(x,y)$ of some graph cannot be ``materialized''. This property is due to the combinatorial nature of the graph.

Let us explain informally what it means to materialize (or \textit{extract}) the mutual information, following the intuition from the seminal papers \cite{gk1973,wyner}. 
For a pair of binary strings $(x,y)$, the string $z$ materializes the mutual information between $x$ and $y$ if 
\[
C(z) \eqp I(x:y) \text{ and } I(x:y|z) \eqp 0.
\]
In the classical example of $x,y$ such that materialization of the mutual information is possible, we take  three independent random strings $a, b$ and $c$ and set $x := ab$ (the concatenation of $a$ and $b$) and $y := bc$. Then the mutual information between $x$ and $y$ is \emph{materialized} by $z = b$. 
In more involved cases, $x$ and $y$ may share large mutual information, though we cannot materialize it (even if logarithmic-precision equalities are substituted with substantially rougher approximations).

The example above has the following combinatorial insight. We can consider a bipartite graph where vertices in both sides are binary strings of length $2n$, and the edges are 
the pairs $(ab, bc)$ for all triplet of $n$-bit strings $(a,b,c)$.
It is easy to see that this graph consists of $2^n$ connected components corresponding to different strings $b$. Such a graph has very weak mixing properties, since a random walk started in one connected component can never jump to another one. Thus, the combinatorial properties of this graph are very different from those exhibited by the graph in Example~\ref{e:poly}. The property of well-mixing of the graph in Example~\ref{e:poly} will help us to prove that mutual information in a typical edge of such a graph cannot be materialized. Both statements of Theorems~\ref{th:1} and~\ref{th:2} imply that we cannot have $I(x:y|z) \eqp 0$ unless $C(z)$ is much larger than $I(x:y)$.

The central idea of both proofs can be explained as follows. In the underlying graph, any induced subgraph cannot be ``very dense'' (i.e., a large fraction of vertex pairs cannot be connected by edges) unless it is ``very small'' (i.e., the total number of vertices must be small, at least in one part of the bipartition).  Higher density of edges in a subgraph corresponds to a smaller value of $ I(x : y | z)$, while obtaining a ``small'' subgraph corresponds to making small $C(x | z) $ or $C(y | z)$;   a tighter trade-off between the combinatorial parameters of the subgraph translates into corresponding information-theoretic inequalities. Although the intuition behind the  proofs of Theorems~\ref{th:1} and~\ref{th:2} is similar, the technical implementations are different. Moreover, as we will see later, these arguments are incomparable (neither one reduces to the other).

\subsection{Mixing as a bounded number of common neighbors}

\begin{theorem}\label{th:1} 
Let $G = (L,R;E)$ be the bipartite graph from Proposition~\ref{p:standard-profile} with $d$-bounded common neighborhood on the right;
let $\bar G = (\bar L,\bar R,\bar E)$ be the $m$-private randomness amplified $G$ (see Definition~\ref{d:privateRandomGraph}), and
let $(\bar x,\bar y) = (\langle x, r_x \rangle, \langle y, r_y \rangle) \in \bar E$ be a typical edge of this graph. Then for every bit string $w$ we have
\[
C(y|x) \lep I(w: \bar y  | \bar x)
\text{ or }
C(x|y) \lep I(\bar x: \bar y|w)  + I(w: \bar x | \bar y)+ \log d.
\]

\end{theorem}

\begin{remark}
Theorem~\ref{th:1} applies to the graph from Example~\ref{e:poly} constructed for the field $\mathbb{F}_{2^n}$ of cardinality $2^n$. Our argument follows the ideas of the proof of non-extractability of the mutual information for a pair that consists of a line on a discrete plane and a point on this line proposed by An.~Muchnik, see  \cite{muchnik-common-info,cmrsv}. 
\end{remark}

\begin{proof}
To make the argument easier to follow, we will use the notations from Figure~\ref{figure1} for $(\bar x, \bar y, w)$. For every $w$, a straightforward computation gives   
\[
C(\bar x|w)  \eqp m+2n-\alpha-\beta, \text{ }
I(\bar x:\bar y|w) \eqp n - \beta, \text{ }
C(\bar y | w) \eqp m+(d+1)n-\beta-\gamma
\]
and 
\[
C(\bar x,\bar y|w)\eqp 2m+(d+2)n - \alpha -\beta - \gamma.
\]
Let $\hat G_w=(\bar L_w, \bar R_w,\hat E_w)$ be the  subgraph from Definition~\ref{d:partitionsGraph} (b). Using Lemma~\ref{l:standard-upper-bound}, we get that\\
\[
|\bar L_w| \le 2^{m+2n - \alpha - \beta+O(\log N)}  \text{ and } |\bar R_w| \le 2^{m+(d+1)n -\beta - \gamma+O(\log N)}.
\]

\smallskip

\noindent
\emph{Observation 1: on the maximal degree in $\hat G_w$.} $\triangleright$
By definition of $\hat G_w$, for every $\bar x' \in \bar L_w$ and for every $\bar y'\in \bar R$ we have
\begin{equation}
\label{eq:sec-3.1-1}
C(\bar x'|\bar y', w) \le \ell,
\end{equation}
where $\ell = C(\bar x|\bar y, w)$. For fixed $\bar y$ and $w$, there are less than $2^{\ell+1}$ different $\bar x'$ verifying~\eqref{eq:sec-3.1-1}. 
Therefore, every vertex in  $\bar R_w$ in graph $\hat G_w$ is incident to less than
\[
2^{C(\bar x|\bar y, w)+1} 
\le 
2^{m+n-\alpha + O(\log N)}
\]
edges in $\hat E_w$. $\triangleleft$

\smallskip

\noindent
\emph{Observation 2: on the cluster structure of $\bar G$.} $\triangleright$
In the graph $\bar G$, every vertex in $\bar L$, $(x , r_x)$ belongs to a \emph{cluster} of $2^m$ vertices of the form $(x , r'_x)$ (same point $x$ but any string $r'_x$) that all have exactly the same neighbors. The same is true for vertices in $\bar R$. Such a cluster of size $2^m$ vertices in $\bar L$ (or, respectively, in $\bar R$) corresponds to a single point $x$ in $L$ (respectively one polynomial $y$ in $R$), and if $(x,y)$ is an edge in the the original graph $G$, then the two associated clusters in $\bar G$ form a complete bipartite graph $ K_{2^m, 2^m}$. $\triangleleft$

\smallskip

In what follows we will select in $\bar R_w$ a subset of $k$ vertices $\bar y_1,\ldots, \bar y_k$ so that for all $\bar y_i = (y_i, r_i)$ the components $y_i$ are pairwise different, i.e., we take all vertices from a different clusters. 
In this case, the number of common neighbors for $(y,r)$ and $(y',r')$ (with $y \neq y'$ and any $r, r'$) is at most $d\cdot 2^m$. Indeed, in the original graph $G$ there are at most $d$ different $x \in L$ incident at once to $y$ and $y'$, and, for each such $x$, there are $2^m$ components $r_x$ so that $(x,r_x) \in \bar L$ is connected by an edge at once with $(y,r)$ and $(y',r')$. For our purposes, we  want these $k$ vertices to have maximal possible degrees in $\hat G_w$ (so that the \emph{average} degree of the $k$ selected vertices is also at least $2^{m+n-\alpha - O(\log N)}$). 

\smallskip

\noindent
\emph{Observation 3: on average degree in $\hat G_w$} $\triangleright$
On the one hand, each vertex in $\bar R_w$ is incident to at most 
\[
2^{m+n-\alpha + O(\log N)}
\]
edges in $\hat E_w$, see Observation~1 above. Besides, we know that vertices in $\bar R$ are grouped in clusters of size $2^m$ whose vertices all have the same neighbors, see Observation~2 above. Hence the neighbors of a clusters are the same as the neighbors of only one vertex in the cluster. It follows that for every cluster in $\bar R_w$ the total number of incident edges in $\hat G_w$ is at most $2^m$ times the maximal degree, hence 
\[
2^{2m+n-\alpha + O(\log N)}
\]
edges. On the other hand, Lemma~\ref{l:induced-edges} guarantees that
\[
|\hat E_w| \ge 2^{2m+(d+2)n - \alpha - \beta - \gamma - O(\log N)}.
\]
Therefore, we can find at least                                

\[
\frac{|\hat E_w|}{2^{2m+n-\alpha + O(\log N)}} \ge 2^{(d+1)n -\beta - \gamma -O(\log N)}
\]
vertices  $\bar y_1,\ldots, \bar y_k\in \bar R_w$ such that 
\begin{enumerate}
\item[(i)] all these vertices belong to pairwise different clusters,
\item[(ii)] the average degree among  $\bar y_1,\ldots, \bar y_k$ is at least $2^{m+n-\alpha - O(\log N)}$.
\end{enumerate}
To this end, we  select at first $k \le 2^{(d+1)n -\beta - \gamma - O(\log N)}$ clusters with maximal possible average degree, and then pick in each of the selected clusters a vertex $\bar y_i$ with maximal possible degree.
(A more precise choice of $k$ is discussed later.)
$\triangleleft$

\smallskip

Now, having selected vertices  $\bar y_1, \dots, \bar y_k \in \bar L_w$ with average degree $\ge 2^{m+n-\alpha - O(\log N)}$,
we count all neighbors of these vertices on the left side of the graph. We denote $\Gamma(v)$ be the set of neighbors of the vertex $v$ in $\hat G_w$. Then
\[
|\Gamma(\bar y_1) \cup \ldots \cup \Gamma(\bar y_k) | \le |\bar L_w| \le 2^{m+2n-\alpha-\beta + O(\log N)} 
\]
(the last inequality follows from Lemma~\ref{l:standard-upper-bound}).

All of the selected vertices from $\bar R_w$ belong to different clusters and, therefore, have $(d\cdot 2^m)$-bounded common neighborhood.
Hence, for every $1\le i < j \le k$, we have $|\Gamma(\bar y_i) \cap \Gamma(\bar y_j)| \le d \cdot 2^m$.
Using the inclusion-exclusion principle, we obtain
\[
|\Gamma(\bar y_1) \cup \ldots \cup \Gamma(\bar y_k)|
\ge \sum_{i=1}^k |\Gamma(\bar y_i)| - \sum_{1\le i<j\le k} |\Gamma(\bar y_i) \cap \Gamma(\bar y_j)|
\]
\[
\ge k\cdot 2^{m+n-\alpha - O(\log N)} - {k \choose 2} d \cdot 2^m.
\]
Thus, 
\[
 k\cdot 2^{m+n-\alpha - O(\log N)} - O(d \cdot 2^m k^2) \le  2^{m+2n-\alpha-\beta + O(\log N)}.
\]
which can be rewritten as 
\begin{equation}
\label{eq:iep}
 k\cdot 2^{m+n-\alpha - O(\log N)}  \le  O(d \cdot 2^m k^2) + 2^{m+2n-\alpha-\beta + O(\log N)}
\end{equation}

We are looking now for a number $k$ that contradicts this inequality. We denote $r = \log k$, $s = \log d$.
To get the contradiction, both terms in the right-hand side  \eqref{eq:iep} need  to be strictly dominated by the term in the left-hand side of this inequality. 
This happens when both inequalities
\[
\left[
\begin{array}{rcl}
r + m + n - \alpha & \lep& s + 2r + m, \\ % \iff n - \alpha - s \gep r.
r + m + n - \alpha &\lep& m+2n-\alpha-\beta
\end{array}
\right.
\]
are false. To this end, we need to find $r$ such that the following two inequalities hold:
  \[
 \left\{
 \begin{array}{rcl}
 n - \alpha -s & \gg&  r , \\ % \iff n - \alpha - s \gep r.
  r    &\gg& n -\beta.
 \end{array}
 \right.
 \]
 In other words, we get a contradiction if we can find $r$ verifying 
 \begin{equation}
 \label{eq:cont}
 n -\beta \ll r \ll  n - \alpha -s.
 \end{equation}
 What can prevent us from taking such an $r$ and getting a contradiction? There are only two reasons why this might be impossible:
 \begin{itemize}
 \item{} [Reason 1]:  there is no positive gap between the right-hand side and the left-hand side in \eqref{eq:cont}, i.e.,
 \[
 n - \alpha -s \lep  n -\beta,
 \]
 which is equivalent to 
 \(
n  \lep   n -\beta + \alpha + s  .
 \)
 \item{} [Reason 2]:  the left-hand side of  \eqref{eq:cont} is 
 too large, and we are unable to find a sufficient number of points $\bar y_j$ in $\bar R$ that belong to pairwise different clusters, i.e., 
 \[
 k = 2^r \ge 2^{ n -\beta } \ge | R| = 2^{(d+1)n -\beta -\gamma + O(\log N)}.
 \] 
This condition implies 
 \(
 dn \lep  \gamma  . 
 \)
 \end{itemize}
Therefore, we avoid a contradiction only if one of these two conditions is verified:
\[
n \lep n - \beta + \alpha + s
\]
(which means that $n \eqp C(x|y)  \lep n - \beta + \alpha + s \eqp I(\bar x: \bar y|w)  + I(\bar x : w | \bar y)+ \log d$), 
or 
\[
dn \lep \gamma
\]
(which means that  $C(y|x) \lep I(\bar y : w | \bar x)$).
This concludes the proof of the theorem.
\end{proof}

\subsection{Mixing property as a large spectral gap}

In this section we prove a statement that is slightly more general than  \cite[Theorem~5]{geoffroy}.

\begin{theorem}\label{th:2}
Let $G = (L,R;E)$ be a bipartite graph with parameters as in Proposition~\ref{p:standard-profile} that is 
a good spectral expander. Let $\bar G = (\bar L, \bar R; \bar E)$ be an $m$-private randomness amplified $G$, 
and let $(\bar x,\bar y) = (\langle x, r_x \rangle, \langle y, r_y \rangle)$ be a typical edge of $\bar G$. Then, for any $w$ such that $I(w:\bar y|\bar x) \eqp 0$ we have 
\[
I(w: \bar x :\bar y) \lep 0 \text{ or } I(w:\bar x|\bar y) + I(w: \bar y | \bar x) \gep n.
\]
\end{theorem}
\begin{remark}
Theorem~\ref{th:2} applies to the graph from Example~\ref{e:poly}. 
\end{remark}
\begin{proof}

Similarly to the proof of Theorem~\ref{th:1},
we define $\bar G_w=(\bar L_w, \bar R_w; \bar E_w)$ as the induced subgraph in $\bar G$ that contains edges that look like $(\bar x, \bar y)$ when is $w$ known (Definition~\ref{d:partitionsGraph}). We use again the notation in Figure~\ref{figure1}. 
From Lemma~\ref{l:standard-upper-bound} we obtain 
\[
|\bar L_w| \le 2^{m+2n-\alpha-\beta \pm O(\log N)} \text{ and } |\bar R_w| \le 2^{m+(d+1)n-\beta  - \gamma \pm O(\log N)} ,
\]
and from Lemma~\ref{l:induced-edges} we obtain 
$|E_w| \ge 2^{2m+(d+2)n-\alpha-\beta  - \gamma \pm O(\log N)}$.

\begin{lemma}[see \cite{mixing-lemma} for the proof]\label{l:expander-mixing-lemma} 

If $\lambda_2$ is the 2nd eigenvalue of a bipartite graph $G=(L,R;E)$, then
\[
\frac{|E_w|}{|E|} \le \frac{ |L_w| \cdot |R_w|}{|L|\cdot |R|} + \lambda_2 \frac{ \sqrt{|L_w| \cdot |R_w|}}{|E|}.
\]
\end{lemma}

For the original graph $G$ we have $\lambda_2 \le \sqrt{2^{dn}}$ (see Lemma~\ref{p:spectrum} and its proof in the appendix section). The eigenvalues of the complete bipartite graph are the numbers $\pm 2^m$ and $0$ (with a high multiplicity). The eigenvalues of a tensor product of two graphs are  pairwise products of the eigenvalues of both tensor components. Therefore, the second eigenvalues of $\bar G$ is $O\left( 2^m \cdot \sqrt{2^{dn}} \right)$.

In our setting, Lemma~\ref{l:expander-mixing-lemma} rewrites to 
\begin{equation}
\label{eq:mixing-lemma}
\begin{array}{rcl}
|\bar E_w| &\le& \frac{|\bar L_p| \cdot |\bar R_p|}{2^{n}} + 2^{m + \frac{dn}{2}} \cdot \sqrt{|\bar L_p| \cdot |\bar R_p|}
\end{array}
\end{equation}
On the right side of \eqref{eq:mixing-lemma} we have a sum of two terms. 
We analyze separately two cases: when the first of these terms is greater than the second one, and when the second term is greater than the first one.

\emph{Case 1:} Assume that in the right-hand side of \eqref{eq:mixing-lemma} the first term dominates. Then  \eqref{eq:mixing-lemma} rewrites to
\[
|\bar E_w| \le O\left(\frac{ |\bar L_w| \cdot |\bar R_w|}{2^{n}} \right).
\]
Combining this with Lemma~\ref{l:standard-upper-bound} we obtain
\[
2m+(d+2)n -\alpha - \beta  - \gamma \lep m + 2n -\alpha - \beta + m + (d+1)n - \beta - \gamma  - n,
\]
which implies that $\beta \lep 0$. This means that the mutual information between $\bar y$ and $w$ is negligibly small. Thus, we get the first clause of the alternative in the conclusion of the theorem.

\smallskip

\emph{Case 2:} Assume that in the right-hand side of  \eqref{eq:mixing-lemma} the second term dominates.
Then  \eqref{eq:mixing-lemma} rewrites to 
\[
|\bar E_w| \le O\left( 2^{m + \frac{dn}{2}} \sqrt{ |\bar L_w| \cdot |\bar R_w| }  \right),
\]
which gives
\[
2m + (d+2)n -\alpha - \beta  - \gamma\lep m + \frac{dn + m + 2n-\alpha-\beta + m + (d+1)n - \beta - \gamma}{2}
\]
and results in $n \lep \alpha  + \gamma$. So we come to the second clause of the alternative in the conclusion of the theorem.
\end{proof}

The following corollary is a particular case of Theorem~\ref{th:2} that will be used in one of our applications (see Proposition~\ref{prop:4}). It corresponds to the setting without private randomness, and basically says that any string $w$ giving more than $C(x|y)$ bits of information on $x$, together with $y$ is enough to compute $x$. This will be useful to show that, in some cases, every message will learn something new to the receiver (this is not what we want from a ciphered message).

\begin{corollary}
\label{cor:1}
If we let $m = 0$ (hence $x = \bar x$, $y = \bar y$) and add to the conditions of Theorem~\ref{th:2} that $I(w:x) \gep n$ and $I(w:y |x) \eqp 0$, then we have $C(x|y,w) \eqp 0$ (i.e. all the missing information from $x$ is revealed to the holder of $y$ via $w$).
\end{corollary}
\begin{proof}
Due to Theorem~\ref{th:2}, we need to consider two  cases: $I(w:x:y) \lep 0$ and $I(w:x|y) + I(w:y|x) \gep n$.

\noindent \emph{Case 1}: $I(w:x:y) \lep 0$.  We combine this inequality with the condition $I(w:x) \gep n$ and obtain 
\[
I(w:x|y) \eqp I(w:x) - I(w:x:y) \gep I(w:x) 
\gep n.
\] 
One the other hand, 
\begin{equation}
\label{eq:I(w:x|y)=n}
C(x|w,y)  \eqp C(x|y) - I(w:x|y) \eqp n -  I(w:x|y) \eqp 0.
\end{equation}

\medskip
\noindent \emph{Case 2}: $I(w:x|y) + I(w:y|x) \gep n$. Since we assume that $I(w:y |x) \eqp 0$, we get $I(w:x|y) \eqp n$. Therefore, we can apply again \eqref{eq:I(w:x|y)=n}.

\end{proof}

\begin{remark}
\label{rem:relativization}
The proofs of Theorem~\ref{th:1} and Theorem~\ref{th:2} relativizes. In fact, all of our statements remain valid if we relativize all terms of Kolmogorov complexity conditional on a string of public randomness $r$. This is true if for every string $s$ used in a proof, we have $I(r:s)\eqp 0$. This happens with overwhelming probability if $r$ is sampled at random independently from any string $s$ used in the proofs.
We present our proofs and statements without $r$. But every step of the arguments trivially relativizes conditional on $r$ assuming that $C(s|r) \eqp C(s)$. We only need to add routinely the random bit string $r$ to the condition of all terms with Kolmogorov complexity appearing in the proof. Such a public randomness may be used in the communication protocols to produce random hash values (for instance, a linear hash function) of the inputs to send only one part of it.

\end{remark}

\section{Linear information inequalities  for graphs with  mixing properties} 
\label{s:inequalities}

The statements and the schemes of the proof of Theorem~\ref{th:1} and Theorem~\ref{th:2} are  similar. One can ask how the  techniques based on bounded size common neighborhood and on a spectral bound relate with each other. It turns out that none of these two methods is superior to the other. These techniques can give incomparable implications  even when being applied to one and the same graph. We illustrate this observation in the next proposition.

\begin{theorem}
\label{th:rustam}
Let $G=(L,R;E)$ be a bipartite graph from Proposition~\ref{p:standard-profile} for $d=1$ 
($L$ consists of points on the affine plane and $R$ consists of  polynomials of degree at most one, which 
represent non-vertical straight lines on the affine plane). Let $(x,y)$ be a typical edge of this graph (so we have $C(x)\eqp C(y)\eqp 2n$ and $I(x:y)\eqp n$).

\noindent
(i) If $G$ has $1$-bounded common neighborhood on the right, then for these $x,y$ and for all $w$ we have
\[
I(w:x,y) \lep 2I(w:x|y) + 2I(w:y|x).
\]
Moreover, for all $m$-bit strings $r_x$ and $r_y$ such that 
\[
C(x,y,r_x,r_y)\eqp C(x,y) + C(r_x)+C(r_y)
\] 
we have
\begin{equation}\label{eq:comb}
I(w:x,y,r_x,r_y) \lep 2I(w:x,r_x|y,r_y) + 2I(w:y,r_y|x,r_x).
\end{equation}

\noindent
(ii) If $G$ is a bipartite spectral expander, then for these $x,y$ and for all $w$ such that $C(w)\le n$ we have
\[
I(w:x,y) \lep I(w:x|y) + I(w:y|x).
\]
Moreover, for such a $w$ and  for all $m$-bit strings  $r_x$ and $r_y$ such that 
\[
C(x,y,r_x,r_y)\eqp C(x,y) + C(r_x)+C(r_y)
\] 
we have
\begin{equation}\label{eq:spec}
I(w:x,y,r_x,r_y) \lep I(w:x,r_x|y,r_y) + I(w:y,r_y|x,r_x).   
\end{equation}
\end{theorem}
The claims (i) and (ii) of the theorem are incomparable: 
the second one gives a sharper inequality but the first one applies to all $w$ without restrictions. Both claims apply to the graph  from Example~\ref{e:poly} with $d=1$.

\begin{remark}
Claim~(i) of Theorem~\ref{th:rustam} without \emph{moreover} part is a slightly stronger version of \cite[theorem~4]{cmrsv} saying (for the same pair $(x,y)$) that $C(w)\lep 2C(w|x)+2C(w|y)$ for all $w$.
\end{remark}

\begin{proof}
The proof of Theorem~\ref{th:rustam} essentially uses the statement of Theorem~\ref{th:1} and Theorem~\ref{th:2}. We will prove directly the stronger \emph{moreover} claims of the proposition. 
Let $\bar G=(\bar L,\bar R,\bar E)$ be an $m$-private randomness amplified $G$,
see  Definition~\ref{d:privateRandomGraph}.

\medskip
We begin with a proof of Claim (i). First of all, we rewrite  \eqref{eq:comb} as
\[
I(w:\bar x,\bar y) \lep 2I(w:\bar x|\bar y) + 2I(w:\bar y|\bar x).
\]
We apply Theorem~\ref{th:1} to $\bar x, \bar y$ and $w$ and obtain
\[
C(y|x) \lep I(\bar y:w|\bar x) \text{ or } C(x|y) \lep I(\bar x:\bar y|w) + I(\bar x:w|\bar y),
\]
Now we need consider two cases.

\emph{Case~1.} Suppose first that 
\[
 C(y|x)\lep  I(\bar y : w|\bar x),
\]
which means here that $I(\bar y:w|\bar x)\gep n$. Since $I(w:\bar x:\bar y) \lep I(\bar x:\bar y) \eqp n$, we have
\[
\begin{array}{rcl}
I(w:\bar x,\bar y)  \eqp I(w:\bar x|\bar y) + I(w:\bar y|\bar x) + I(w:\bar x:\bar y)& \lep& I(w:\bar x|\bar y) + I(w:\bar y|\bar x) + n\\
&& I(w:\bar x|\bar y) + 2 I(w:\bar y|\bar x), 
\end{array}
\]
and \eqref{eq:comb} follows immediately.
\medskip

\emph{Case~2.} Alternatively, suppose that we have 
\[
n \eqp C(x|y) \lep I(\bar x:\bar y|w) + I(w:\bar x|\bar y),
\]
Again, we observe that  $I(w:\bar x,\bar y) \eqp I(w:\bar x|\bar y) + I(w:\bar y|\bar x) + I(w:\bar x:\bar y)$. The sum of these inequalities gives 
\[
\begin{array}{rcl}
I(w:\bar x,\bar y) + n &\lep& I(w:\bar x|\bar y) + I(w:\bar y|\bar x) +  I(\bar x:\bar y:w)   \\
&& {} + I(\bar x: \bar y|w ) +  I(w:\bar x|\bar y)
\end{array}
\]
Using that $I(w:\bar x:\bar y) + I(\bar x:\bar y|w) = I(\bar x:\bar y) = n$, we obtain
\[
I(w:\bar x,\bar y) \lep 2I(w:\bar x|\bar y) + I(w:\bar y|\bar x),
\]
which implies \eqref{eq:comb}.

\medskip

Now we prove Claim~(ii) of the theorem. 
With our notation from above, Inequality \eqref{eq:spec} rewrites to 
\[
I(w: \bar x,\bar y) \lep I(w: \bar x| \bar y) + I(w: \bar y| \bar x).
\]
From Theorem~\ref{th:2} it follows 
\[
I(w:\bar x :\bar y) \lep 0 \text{ or } I(w:\bar x|\bar y) + I(w: \bar y | \bar x) \gep n.
\]

\emph{Case~1.} Assume that $I(w:\bar x :\bar y) \lep 0$. Then

\[
I(w: \bar x,\bar y) \eqp I(w: \bar x| \bar y) + I(w: \bar y| \bar x) + I(w:\bar x :\bar y)
\lep I(w: \bar x| \bar y) + I(w: \bar y| \bar x) ,
\]
and we obtain exactly \eqref{eq:spec}.

\medskip

\emph{Case~2.} Now we assume that $I(w:\bar x|\bar y) + I(w: \bar y|\bar x) \gep n$. 
We keep in mind that for all strings 
\[
 I(w: \bar x| \bar y) \lep C(w)
 \]
and use the assumption
$
C(w) \lep n.
$
Combining all these inequalities we obtain \eqref{eq:spec}.

\end{proof}

\section{Application I : Communication complexity for secret key agreement}

\label{s:app1}

We assume that Alice and Bob get inputs $x$ and $y$ respectively and agree on a secret key $z$, following some fixed communication protocol (possibly with public and private sources of randomness accessible to the Eavesdropper), see \cite{andrei-marius,emirhan,geoffroy} for a detailed discussion of the secret key agreement in the framework of Kolmogorov complexity.

First of all, we recall a lemma saying that for every communication protocol the \emph{external communication complexity} (information learned on the inputs of Alice and Bob by an external observer who can access the transcript of the protocol) is greater than the \emph{internal information complexity} (the amount of information learned during the communication by Alice on Bob's input plus the amount of information learned by Bob on Alice's input):

\begin{lemma}[see \cite{andrei-marius}]\label{l:triple-info}
If $t$ is a transcript of a communication protocol on inputs $x$ (given to Alice) and $y$ (given to Bob), then 
\[
I(t:x:y) \eqp  I(t:x,y) - ( I(t:x|y) + I(t:y|x) ) \gep 0.
\]
\end{lemma}

\noindent A version of this lemma for Shannon's entropy is well known in communication complexity, see, e.g., \cite{info-complexity}.

In what follows we explain the proof for communication protocols with private randomness and without public randomness. The argument easily adapts to the setting with a public source of randomness, see Remark~\ref{rem:relativization}.

\begin{theorem}\label{th:secret_key_agreement}
Let $G = (L, R; E)$ a bipartite graph that has bounded common neighborhood (we will take the graph from Example~\ref{e:poly}) and $\bar G = (\bar L, \bar R; \bar E)$ be
the $m$-private randomness amplified $G$. 
Let $(\bar x, \bar y) = (\langle x, r_x \rangle, \langle y, r_y \rangle)$ be a typical edge of $\bar G$. 

Suppose that Alice is given $\bar x$ and Bob is given $\bar y$
and they want to agree on a secret key $z$ of maximal complexity $I(\bar x:\bar y)$. Then, achieving such a secret key agreement protocol implies one of the two following situations:
\begin{itemize}
    \item Alice needs to send $C(x|y)$ bits and Bob nothing
    \item Bob needs to send $C(y|x)$ bits and Alice nothing.
\end{itemize} 
\end{theorem}

\begin{proof}

Consider $x$ to be a point in the finite affine plane of size $2^n \times 2^n$ and $y$ to be a polynomial of degree $d = O(n)$ that goes through $x$. Adding the private randomness $r_x$ and $r_y$, we get that $(\bar x, \bar y)$ is a typical edge of the graph considered in Example~\ref{e:poly} with private randomness (Definition~\ref{d:privateRandomGraph}).

Suppose that Alice and Bob agreed on a secret key $z$ of size $I(\bar x:\bar y)$ using a transcript $t$ that is the concatenation of all the messages of both participants that they sent to each other to achieve the protocol. We have then the following relations:

\[
\begin{array}{rccl}
I(z:t)&\eqp& 0 &\text{(information theoretic security, the transcript reveals no information on the key)}\\
C(z|t,\bar x) &\eqp& 0 &\text{(Alice knows the key as she knows $\bar x$ and the transcript $t$)}\\
C(z|t,\bar y) &\eqp& 0 &\text{(Bob knows the key  as he knows $\bar y$ and the transcript $t$)}\\
I(\bar x: \bar y: t ) &\gep& 0 &\text{(From Lemma~\ref{l:triple-info})}
\end{array}
\]
For our purposes, we study the quantity
\[
I(\bar x: \bar y|z,t) \eqp C(\bar x,z,t) + C(\bar y,z,t) - C(\bar x,\bar y,z,t) - C(z,t).
\]
Since $I(z:t) \eqp 0$ we have $C(z,t) \eqp C(z) + C(t)$. Moreover, $z$ can be computed from $(t,\bar y)$ or $(t,\bar x)$, thus $C(\bar x,z,t) \eqp C(\bar x,t)$ and $C(\bar y,z,t) \eqp C(\bar y,t)$, which gives

\begin{equation}
\label{eq:sec5-1}
I(\bar x: \bar y|z,t) \eqp C(\bar x,t) + C(\bar y,t) - C(\bar x,\bar y,t) - C(t) - C(z) \eqp I(\bar x : \bar y |t) - C(z).
\end{equation}
Since we have  $I(\bar x:\bar y:t) \gep 0$, the value $I(\bar x : \bar y |t)$ cannot exceed $I(\bar x : \bar y) \eqp C(z)$.
We combine this observation with \eqref{eq:sec5-1} and obtain $I(\bar x: \bar y|z,t) \lep 0$. 
On the other hand,  $I(\bar x: \bar y|z,t)$ cannot be negative. So we get that 
\begin{equation}
\label{eq:sec5-2}
I(\bar x: \bar y|z,t) \eqp 0.
\end{equation}

We now can apply the Theorem~\ref{th:1} to $(\bar x, \bar y)$ and $w = (t, z)$. The theorem claims that 
\[
C(y|x) \lep I(w:\bar y  | \bar x)
\text{ or }
C(x|y) \lep I(\bar x: \bar y|w)  + I(\bar x : w | \bar y)+ \log d.
\]

\emph{Case 1:} $I(\bar y : z,t | \bar x) \gep C(y|x)$. Then  we have
\[
dn \eqp C(y|x) \lep I(\bar y : z,t | \bar x) \lep I(\bar y : t | \bar x) + I(\bar y : z | \bar x,t)  \lep I(\bar y : t | \bar x) + C(z | \bar x,t)  \eqp I(\bar y : t | \bar x).
\]
Thus, we obtain  $I(\bar y : t | \bar x) \gep dn$. This means that  the information  revealed by  Bob  on his input is at least $dn$ bits. 

\emph{Case 2:}   $C(x|y) \lep I(\bar x: \bar y|z,t)  + I(\bar x : z,t | \bar y)$. We combine this with \eqref{eq:sec5-2} and obtain 
\[
n \eqp  C(x|y) \lep I(\bar x : z,t | \bar y) =  I(\bar x : t | \bar y) +  I(\bar x : z | \bar y,t) \lep   I(\bar x : t | \bar y) +  C( z | \bar y,t) \lep  I(\bar x : t | \bar y).
\]
Thus, we have $ I(\bar x : t | \bar y) \gep n$. This means that the information revealed by  Alice  on her input is at least $n$ bits.
\end{proof}

\begin{remark}
This statement easily adapts to the case when Alice and Bob do not use private randomness to achieve their protocol by setting $r_x$ and $r_y$ to the empty words. This way, we get $\bar x = x$ and $\bar y = y$ (and $\bar G = G)$. One can go through the proof again doing these substitutions and get to the same conclusion that Alice or Bob need to send all of their information unknown to the other participant to agree on a secret key. Moreover, in both cases (with or without private randomness) the proof relativizes and the statement remains true if the participants (and the eavesdropper) both have access to a \emph{public} source of random bits (see Remark~\ref{rem:relativization}).   
\end{remark}

\begin{remark}
\label{rem:o(n)}
We formulated and proved Theorem~\ref{th:secret_key_agreement} within a logarithmic precision: in the statement of the theorem and in the proof, all equalities and inequalities hold up to an additive term $O(\log n)$, where $n$ is the mutual information between Alice's and Bob's inputs.
This the most natural (and the tightest possible) precision for most information-theoretic properties involving Kolmogorov complexity. However, we could repeat the argument step-by-step with a coarser precision, up to additive term s$O(\sqrt{n})$, or $O(n^{3/4})$, or just up to  $o(n)$. A version of  Theorem~\ref{th:secret_key_agreement}  with a coarser precision might be helpful in some applications, for example, in the interplay between Kolmogorov complexity and Shannon entropy, as we see in the next section.
\end{remark}

\section{Implications for  Shannon's information theory}
\label{s:shannon}

In this section we use Theorem~\ref{th:secret_key_agreement} to prove a similar fact in the  framework of the classical information theory. In fact, this is the original approach to the idea of information-theoretic security for secret key agreement, see  \cite{secret-key-1,secret-key-2}.
We consider the following probabilistic setting.
Let $G=(L, R; E)$ be a graph as in Example~\ref{e:poly} over a field $\mathbb{F}$ o size  $2^n$. We introduce a uniform distribution on $E$. So we sample a random edge in the graph, and denote $({\cal X},{\cal Y})$ the left and the right end of this random edge. We observed in Proposition~\ref{p:standard-profile}  that 
\[
H({\cal X}) = 2n,\ H({\cal Y}) = (d+1)n,\ H({\cal X},{\cal Y}) = (d+2)n
\]
and
\[
I({\cal X}:{\cal Y}) = H({\cal X}) + H({\cal Y}) - H({\cal X},{\cal Y}) = n,
\]
where $H(\cdot)$ stands for Shannon's entropy and $I(\cdot ) $ stands for Shannon's mutual information. 

We study deterministic communication protocol with two participants, Alice and Bob. We provide Alice with this edge's left end (a random variable ${\cal X}$) and Bob with this edge's right end (a random variable ${\cal Y}$). Alice and Bob exchange messages and both compute at the end of their communication some value $z$ in $\{0,1\}^n$, which will be interpreted as a secret key. 

Let us fix any communication protocol. As we run this protocol on a randomly sampled pair of inputs (for Alice and Bob), the resulted secret key becomes a random variable (we denote it $\cal Z$). Besides, the transcript of the protocol (the concatenation of all messages sent by Alice and Bob to each other) also becomes a random variable (we denote it $\cal T$). The aim is to maximize the entropy of $\cal Z$ while the mutual information $I({\cal Z}:{\cal T})$ remains negligibly small.
That is, we want to maximize the entropy of the secret key $\cal Z$ without revealing much information on it to the eavesdropper who can intercept the communication
between Alice and Bob.

In what follows we list several claims known about this setting. The first six claims are previously known results, presented here as motivation for the main result of this section. Only the last one is new, and we prove it at the end of this section.

\begin{claim} 
\label{claim-i}
There exists a communication protocol that allows Alice (given $\cal X$) and Bob (given $\cal Y$) to agree on a secret key $\cal Z$ such that $H({\cal Z}) = n$, and the mutual information between $\cal Z$ and the transcript of the protocol is at most $o(n)$.
\end{claim}
In slightly different setting (for $\cal X$  and $\cal Y$ sampled from i.i.d. sources) this claim was proven in  \cite{secret-key-1,secret-key-2}. 
Information-theoretic secret key agreement was studied, for example, in 
one-shot setup was discussed in \cite{one-shot-key-agreement-1,one-shot-key-agreement-2,one-shot-key-agreement-3,one-shot-ska-2}, see also \cite{andrei-marius,emirhan}. 
For our specific setting, we will give a simple direct proof based on the ideas from  \cite{secret-key-1,secret-key-2}.

\begin{claim} 
\label{claim-ii} 
Moreover, there exists a communication protocol that allows  Alice (given $\cal X$) and Bob (given $\cal Y$) to agree on a secret key $\cal Z$ with $H({\cal Z}) = n$, where
Alice sends to Bob $n$ bits and Bob sends nothing.
\end{claim}
See the proof below.
 
 \begin{claim} 
\label{claim-iii}
Also, there exists a communication protocol that allows Alice (given $\cal  X$) and Bob (given $\cal Y$) to agree on a secret key $\cal Z$ such that $H({\cal Z}) = n$, where
Bob sends to Alice $dn$ bits and Alice sends nothing.
\end{claim}
See the proof below.

\begin{claim} 
\label{claim-iv}
For every $\epsilon>0$, there exists no communication protocol that would allow Alice (given $\cal X$) and Bob (given $\cal Y$) to agree on a secret key ${\cal Z}'$ with entropy  $H({\cal Z}')>(1+\epsilon)n$ such that the mutual information between $Z'$ and the transcript of the protocol is still  $o(n)$. Thus, the key size in Claim~\ref{claim-i},  Claim~\ref{claim-ii}, and Claim~\ref{claim-iii}  is asymptotically optimal.
\end{claim}
In \cite{secret-key-1,secret-key-2} a version of this claim was proven for  $\cal X$  and $\cal Y$ sampled from i.i.d. sources. See also \cite{one-shot-key-agreement-3,one-shot-ska-2} for lower bounds in one-shot setting and 
\cite{andrei-marius} specifically for the case where  $\cal X$ and $\cal Y$ are points and lines on the plane. 

\begin{claim} 
\label{claim-v}
In every communication protocol where Alice and Bob agree on a secret key $\cal Z$ with  $H({\cal Z}) = n - o(n)$, the communication complexity of this protocol 
(i.e., the maximum number of bits communicated, maximized over all inputs) is at least $n - o(n)$. Thus, the communication complexity in the protocol from Claim~\ref{claim-ii} is asymptotically optimal.
\end{claim}
For the proof of this claim see \cite{emirhan}.
 
\begin{claim} 
\label{claim-vi}
Moreover, in every communication protocol where Alice and Bob agree on a secret key $Z$ with  $H(Z) = n-o(n)$, either there exists a pair of inputs such that Alice sends $n-o(n)$ bits of information to Bob, or  there exists a pair of inputs such that Bob sends $dn-o(n)$ bits of information to Alice.
So, in some sense, the protocols from Claim~\ref{claim-ii} and Claim~\ref{claim-iii} cannot be improved.
\end{claim}

\begin{proof}[Sketch of proof of Claim~\ref{claim-i} and Claim~\ref{claim-ii}:] Alice is given a (sampled at random) point $(x_1,x_2) \in \mathbb{F}^2$ on the affine plane. 
Let her send to Bob the first component of this pair, i.e., $x_1$. Having received this message, Bob (who is given a polynomial incident to the point $(x_1,x_2)$) can compute  $x_2$. Then, Alice and Bob use the value $x_2$ as the secret key.

From the construction it follows immediately that the entropy of the key (the value of $x_2$) is equal to $\log |\mathbb{F}|$, and that the mutual information between the key and the transcript (i.e., between the values of $x_1$ and $x_2$ sampled at random) is equal to $0$. 
\end{proof}

\begin{proof}[Sketch of proof of Claim~\ref{claim-iii}:] Bob is given (sampled at random) $d+1$ coefficients of a polynomial 
\[
S(t)= s_0 + s_1t + \ldots + s_d t^d.
\]
Let him send to Alice all coefficients of this polynomial except for $s_0$. Having received this message, Alice (who is given an evaluation point for this polynomial)- computes $s_0$. Then, the value of $s_0$ is taken as the secret key.

It is easy to see that the entropy of the key (the value of $s_0$) is equal to $\log |\mathbb{F}|$, and that the mutual information between the key and the transcript (i.e., between the values of $s_0$ and $(s_1,\ldots,s_d)$) 
is equal to~$0$. 
\end{proof}

\begin{proof}[Sketch of the proof of Claim~\ref{claim-vi}:]
To prove this probabilistic statement, we will use  the standard connection between Kolmogorov complexity and Shannon entropy, see \cite{grunwald} or the textbooks \cite{li-vitanyi,suv} for an elaborate discussion of the interplay between the algorithmic and classical versions of information theory. Technically, it is enough to apply the following lemma.
\begin{lemma}
\label{lemma:shannon-kolmogorov}
Let $A$ be a finite set of cardinality $M$. 
We assume that this set is given explicitly in the sense that there exists a program of size $k$ that produces the list of all elements of $A$ and stops.
Let $\cal U$ be a random variables distributed on $A$ such that $H({\cal U}) \ge \log M - \Delta_1$. Then there exists $\Delta_2 = O(\Delta_1 + k+\log \log M)$ such that 
\[
\pr_{u\leftarrow \cal U}[C(u)< \log M - \Delta_2] < 1/10.
\]
In other words, with high probability Kolmogorov complexity of $u$ is bigger (or not much less) than the Shannon entropy of the distribution.
\end{lemma}
\begin{remark}
This lemma is interesting in the case when $\Delta_2 \ll  \log M $ and, therefore, $H({\cal U}) $ is close to  $\log M$ 
(we achieve the equality $H({\cal U}) = \log M$ only with the uniform distribution on a set of cardinality $M$).
\end{remark}
This lemma follows easily from \cite[Theorem~2.10]{grunwald}, see Appendix for details.
\begin{corollary}
\label{cor:shannon-kolmogorov}
(a) If the list of all elements of a set $A\subset \{0,1\}^N$ can be produced by an algorithm of size $k$,  and there is a probability distribution $\cal A$ having $A$ as its support 
such that $H({\cal A}) \ge \log |A| - \Delta$,
then for a randomly chosen (with the  distribution $\cal A$) element  $u\in A$ we have 
\[
\pr_{u\leftarrow \cal U}\big[    \log |A|  - O(\Delta) - O(\log N)  - O(k)  \le   C(u) \le \log |A| +  O(k) \big] \ge 9/10.
\]
In particular, if $H({\cal A}) \ge  \log |A| - o(N)$  and $k = O(\log N)$, then
\[
\pr_{u\leftarrow \cal U}\left[  \log |A| - o(N) \lep C(u)  \lep \log |A| \right] \ge 9/10.
\]
(b) With the same assumption on $A$ and for any finite set $B$, if $({\cal A}, {\cal B})$ is distributed on $A\times B$
and $H({\cal A} | {\cal B}) \ge \log |A| - \Delta$, 
then  for a randomly chosen pair  $(u,v)\in A\times B$ we have 
\[
\pr_{(u,v)\leftarrow{\cal A}, {\cal B}}\big[    \log |A| - O(\Delta)   - O(\log N)  - O(k)  \le   C(u|v) \le \log |A| +  O(k) \big] \ge 9/10.
\]
In particular, if $H({\cal A} | {\cal B}) \ge \log |A| - o(N)$  and $k = O(\log N)$, then
\[
\pr_{(u,v)\leftarrow{\cal A}, {\cal B}}\left[   \log |A| - o(N)   \lep C(u|v)  \lep \log |A|  \right] \ge 9/10.
\]

\end{corollary}
\begin{proof}
(a) We apply Lemma~\ref{lemma:shannon-kolmogorov} and observe that with  probability $>9/10$ the value of $C(u)$ is sandwiched 
between the lower bound $\log |A| - O(\Delta)- O(\log N) - O(k)$  from the lemma and the trivial  upper bound 
\[
\log |A| + [\text{self-delimiting description of the list of all elements of $A$}].
\]
(b) We apply claim~(a) of the corollary for the conditional distributions on $A$ given ${\cal B} =v$ for each $v\in B$, 
and then take the average over all  $v$ with respect to the distribution $\cal B$.
\end{proof}

Let us fix now a communication protocol in which Alice and Bob agree on a secret key. 
We denote by $\cal Z$ the produced random value of the secret key (a random variable sampled together with $({\cal X},{\cal Y})$)
and by $\cal T$ the transcript of the protocol (also a random variable sampled together with $(X,Y)$).
By the assumption,  $H({\cal Z})\ge n-o(n)$ and $I({\cal Z}:{\cal T}) = o(n)$.

We apply Corollary~\ref{cor:shannon-kolmogorov}~(a) to each of the following random variables: $({\cal X},{\cal Y})$, ${\cal X}$, $\cal Y$, $\cal T$;
we also apply Corollary~\ref{cor:shannon-kolmogorov}~(b) to $({\cal T},{\cal Z})$.  
We conclude that with probability $>5/10$ for a randomly chosen (over the uniform distribution) 
$(x, y)\in E$ the protocol produces a transcript $t$ and the final result $z$ such that 
\begin{itemize}
\item $C(x,y) \eqp 3n$, 
 $C(x) \eqp 2n$, 
 $C(y) \eqp 2n$,  i.e.,  the complexity profile of $(x,y)$ matches \eqref{eq:profile}, 
\item $ n-o(n) \le C(z) \le n + O(\log n)$,  
 \item $C(z) - C(z|t) \le n - (n-o(n)) = o(n)$. 
\end{itemize}
We take such a pair $(x, y)$ and apply Theorem~\ref{th:secret_key_agreement} 
(or, more exactly, its coarser version with $o(n)$ instead of $O(\log n)$ terms, see Remark~\ref{rem:o(n)}).
It implies that for this pair of inputs $(x,y)$  the communication protocol requires that either Alice sends messages with the total length $n-o(n)$ bits,
or Bob sends messages with the total length $dn-o(n)$. This concludes the proof.
\end{proof}

\begin{remark}
We cannot claim that for every communication protocol either Alice sends $n$ bits for all (or almost all)  pairs of inputs $(x,y)$, or Bob sends $dn$ bits for  all (or almost all) pairs of inputs $(x,y)$. In fact, there exist protocols such that in same cases (for some pairs of inputs) Alice sends $n$ bits, while in other cases Bob sends $dn$ bits. In particular, there exist communication protocols where Shannon's entropy of Alice's messages (on a random pair of inputs $({\cal X},{\cal Y})$) is much less than $n$, and Shannon's entropy of Bob's messages (also on a random pair of inputs $({\cal X},{\cal Y})$)  is much less than $dn$. 
But we do claim that in almost all branches of the protocol one of the participants sends the maximal number of bits (about $n$ bits for Alice or about $dn$ bits for Bob).
Such a ``non-convex'' bound is somewhat unusual for Shannon's setting; it cannot be proven directly with the conventional technique of linear inequalities for Shannon's entropy.
\end{remark}

\section{Application II : Impossibility results on secure conditional descriptions}
\label{s:app2}

An.~Muchnik came up in \cite{muchnik-crypto} with the following communication problem. Sender should transmit to Receiver a \emph{clear message} $x$ in an encrypted form assuming that Sender and Receiver share a \emph{secret key} $y$. The encrypted message should not reveal any information on $x$ to Eavesdropper who owns some side information $z$. Muchnik argued that the possibility to find a suitable cipher code (``useful'' for Received but ``useless'' for Eavesdropper)
may depend not simply on the information quantities but on the combinatorial structure of the data sets $x,y,z$.

Proposition~\ref{p:muchnik} below is inspired by \cite[Theorem~5]{muchnik-crypto}. In some respects, our result is weaker: we use a sharp definition of ``usefulness'' and ``uselessness'' while Muchnik's theorem captures a subtle trade-off between smooth versions of these parameters. On the other hand, we use a broader and more natural condition on $x$ and $y$ (we assume independence in the conventional sense, while Muchnik assumed independence with respect to the oracle solving the halting problem for Turing machines). Our argument is purely combinatorial, while Muchnik's proof employed tools from the computability theory.

\begin{definition}
We say that a string (message) $p$ is \emph{useful}  to get a string $x$ given a string $y$ if $C(x|p,y) \eqp 0$, and  $p$ is \emph{useless}  to get a string $x$ given a string  $z$ if $C(x|p,z)\eqp C(x|z)$.
\end{definition}

\begin{proposition}
\label{p:muchnik}
For any integer numbers $d = O(n)$ and $n \ge 0$ and any pair of strings $x,y$ such that $C(x) \eqp C(y) \eqp n$ and $I(x:y) \eqp 0$ there exists a string  $z$ such that 
\[
\begin{array}{l}
C(z)\eqp(d+1)n, \ I(x:z) \eqp I(y:z) \eqp0,
\end{array}
\] 
(with $d = O(n)$) and there is no string $p$ significantly smaller than $(d + 1)n$  (let us say, smaller than $dn + 0.99n$) such that at once\\
(i) $p$ is useful to get $x$ given $y$, and \\
(ii) $p$ is useless to get $x$ given $z$.

\end{proposition}
\begin{proof}
We take $((x,y),z)$ as in  Example~\ref{e:poly} ($z$ is a polynomial of degree $d$ going through the point $(x,y)$). Since  $C(x|y)\eqp n$, for every $p$ that is useful to get $x$ from $y$, we have  $I(p:x,y) \gep n$. Now we assume for the sake of contradiction that $p$ is useless to get $x$ from $z$. This means that $I(p:x|z) \eqp 0$. In this case $I(p:(x,y):z)$ is equal to
\[
I(p:x,y) - I(p:x,y|z) \eqp I(p:x,y) - I(p:x|z) - I(p:y|x,z)\eqp  I(p:x,y) - 0 - 0 \eqp  I(p:x,y)  \eqp n
\]
and
\[
I(x,y:z|p) \eqp I(x,y:z) - I(p:(x,y):z) \eqp n-n \eqp 0.
\]
We apply Theorem~\ref{th:1} to the pair $((x,y), z)$ with the string $p$ (with $m=0$, i.e., without private randomness) and conclude 
\[
C(z|x,y) \lep I(z:p|x,y)
\text{ or }
C(x,y|z) \lep I(x,y:z|p)  + I(x,y:p|y)+ \log d.
\]
Thus, since $\log d = O(n)$ and  $I(x,y:z|p) \eqp 0$, the conclusion of Theorem~\ref{th:1} rewrites to 
\[
dn \lep I(p:z|x,y) \text{ or } n \lep I(p:x,y|z).
\]
The second alternative of the disjunctive is impossible (due to the ``uselessness'' condition), so the first one must be true. Therefore,
\[
I(p:x,y,z) \eqp I(p:z|x,y) + I(p:x,y) \gep dn +n,
\]
which is impossible if the length of $p$ is less than $dn + 0.99n$. Thus, the assumption of uselessness of $p$ for $z$ was false. 
\end{proof}

In Proposition~\ref{p:muchnik} it is crucial that the size of the message $p$ significantly smaller than $(d+1)n$ (the size of $z$). For $(x,y,z)$ used in the proof of this proposition it is easy to find $p$ of length $(d+1)n$ that is at once useful to get $x$ from $y$ and useless to get $x$ from $z$. Indeed, it is enough to let $p=z$. 
We can prove a much stronger statement if we assume that $p$ must be independent of $z$ conditional on $(x,y)$. (This condition is natural if $p$ is sampled by Sender who knows the clear message $x$ and the secret key $y$ but not the side information $z$ known to Eavesdropper). In the next proposition we allow $p$ to be of size $\mathrm{poly}(n)$ (possibly much longer than $x$, $y$, and $z$). Moreover, we show that if $p$ is useful to get $x$ from $y$, then it is useful to get $x$ from $z$ (this condition is much stronger than the negation of uselessness).
\begin{proposition}\label{prop:4}
For any integer numbers $d \ge 0$ and $n\ge 0$ and any pair of strings $(x,y)$ such that $C(x) \eqp C(y)\eqp n$ and $I(x:y)\eqp 0$, there exists a string  $z$ such that
\[
\begin{array}{l}
C(z)\eqp(d+1)n \text{ and } I(x:z) \eqp I(y:z) \eqp0,
\end{array}
\] 
and for every $p$ of length $\mathrm{poly}(n)$ satisfying $I(p:z|x,y)\eqp 0$, if $p$ is useful to get $x$ from $y$, then $p$ is useful to get $x$ from $z$.
\end{proposition}
\begin{proof}
Again, we take $((x,y), z)$ from Example~\ref{e:poly}: $z$ is a binary string representing a polynomial of degree $d$ over $\mathbb{F}_{2^n}$, and $x,y$ are two $n$-bit binary strings representing elements of the field $\mathbb{F}_{2^n}$ (we require that the point $(x,y)$ belongs to the graph of the polynomial represented by $z$). Similarly to the proof of Proposition~\ref{p:muchnik}, the usefulness implies $I(p:x,y) \gep n$. Since we assumed that $I(p:z |x,y) \eqp 0$, we can apply Corollary~\ref{cor:1}, which gives $C(x,y|p,z) = 0$, and we are done.
\end{proof}

Some closely related open questions were formulated in \cite{vereshchagin_crypto} (section 5). In particular the first open question can be answered using our previous proposition. We can first recall Theorem 2 from \cite{muchnik-crypto}:

\begin{theorem}[see \cite{muchnik-crypto}, Theorem 2]
Let $x$, $y$, and $z$ be three strings. Then there exists a string $p$ such that

\noindent (i) $C(x|y, p) \eqp 0$;

\noindent (ii) $C(x|z, p) \eqp \min\{C(x|z), C(y|z)\}$.
\end{theorem}

In other words, this claims that there exists a message that is \emph{useful} to reconstruct $x$ using $y$ and \emph{useless} to produce $x$ using $z$. In the proof of this theorem, a string $p$ of exponential size (in the size of $x$, $y$ and $z$) is used. The first open question in \cite{vereshchagin_crypto} asks whether this is also true for a string $p$ of polynomial size.

We can show that this is false when we add the condition that $p$ is independent from $z$ conditional to $x$ and $y$: $I(p:z|x,y) \eqp 0$ (again, this is a natural condition since there is no reason why the sender should know the input of the adversary). By taking $x$, $y$, and $z$ as in Proposition~\ref{prop:4}, a direct application of this statement allows to conclude that there is no string $p$ of polynomial size such that $p$ is independent from $z$ (conditional on $\langle x, y \rangle$) and at the same time useful to get $x$ from $y$ and useless to get $x$ (or $y$) from $z$. Actually every string $p$ useful to get $x$ from $y$ will be also useful to get $x$ from $z$, which compromises the security of the communication.

\appendix

\section{Proof of Lemma~\ref{p:spectrum}}

\begin{proof}
Let $q$ be the size of the field $\mathbb{F}_n$ and let
\[
A =\left(
\begin{array}{cc}
0 & M\\
M^\top &0
\end{array}
\right)
\]
be the adjacency matrix of the graph from Example~\ref{e:poly}. Recall that, $|L|  = q^2$ and $|R| = q^{d+1}$. The degree on the left is $D_L  = q^d$ (numbers of polynomials of degree $d$ that go through a specific point); the degree on the right is $D_R  = q$ (number of points in the graph of a polynomial). It is enough to estimate the eigenvalues of $A^2$. To this end, we need to estimate the eigenvalues of $M\cdot M^\top $, 
which is the matrix of paths of length $2$ in the graph, starting and finishing in $L$ (the set of points). Starting at some $x\in L$, we can go to a $y\in R$ (which corresponds to some polynomial whose graph goes through the point $x$) 
and then either come back to the same $x$, or end up in a different $x'\in L$. For a fixed $x$, the number of paths 
\[
x \to y \to x
\]
is equal to $D_L = q^{d}$ (any $y$ connected to $x$ can serve as the middle point of the path). We now estimate the number of paths
\[
x \to y \to x'.
\]
For $x \neq x'$. Let us denote $x = (x_1, x_2)$ and $x' = (x_1', x_2')$. If $x_1 \neq x_1'$, (that is, the horizontal coordinates are different), the number of paths 
is equal to the number of polynomials $y$ that go through both $x$ and $x'$, which is $q^{d-1}$.

If $x_1 = x_1'$ (that is, the horizontal coordinate is the same) the numbers of paths is zero. Let $S$ be the matrix of size $q^{2} \times q^{2}$ such that $S_{x,x'} = 1$ if $x_1 = x_1'$ and zero otherwise. Hence $S$ is the matrix of points that share the same first coordinate. It is not hard to see that $S = I_{q} \otimes J_{q}$ (where $I_i$ and $J_i$ are respectively the identity matrix and the all one matrix of size $i$). Indeed, $S$ is the matrix with $q$ all ones blocks of size $q \times q$ in the diagonal and zero elsewhere. Hence, we get
\[
M\cdot M^\top =
 q^{d}\cdot I_{q^{2}} +  q^{d-1} \cdot (J_{q^{2}} - S). 
\]
The identity matrix has the eigenvalue $1$ of multiplicity $|L| = q^{2}$, and the all-ones matrix has the eigenvalue $|L|$ of multiplicity $1$ and the eigenvalue $0$ of multiplicity $(|L|-1)$.  These matrices obviously have a common basis of eigenvectors. The matrix $S$ has eigenvalue $q$ with multiplicity $q$ and zero with multiplicity $(q-1)q$.

Let $v_1 = (1, \dots 1)$ be a vector of dimension $q^2$, the eigenvector associated to the largest eigenvalue of $M\cdot M^\top$, and let $v$ the vector associated with the second largest eigenvalue of $M\cdot M^\top$. This matrix is symmetric, hence $v_1$ and $v$ are orthogonal. Moreover, $v_1$ is also the eigenvector associated to the largest eigenvalue of both $J_{q^2}$ and $S$. Hence, $I_{q^{2}} \cdot v = v$, $J_{q^{2}} \cdot v = 0 \cdot v$. For $S \cdot v$, we have the choice between $0 \cdot v$ or $q\cdot v$. The former case results in a larger eigenvalue (for $M\cdot M^\top$) than the latter (but still smaller than the largest eigenvalue $q^{d+1}$). Therefore, 

\[
M\cdot M^\top \cdot v =
q^{d}\cdot I_{q^{2}} \cdot v +  q^{d-1} \cdot (J_{q^{2}} - S) \cdot v
= q^{d} \cdot v +  q^{d-1} \cdot (0 \cdot v - 0 \cdot v) = q^d \cdot v.
\]

The eigenvalues of $A$ are the square roots of that of $M\cdot M^\top$.
\end{proof}

\section{Proof of Lemma~\ref{lemma:shannon-kolmogorov}}

\begin{proof}
Let us introduce some notation. Let $m=\log M$. Since set $A$ is given by a program of size $k$, for every $u\in A$ we have
\[
C(u) \le \log |A| +  [\text{self-delimiting description of the list of all elements of $A$}]  = m + O(k).
\] 
Further, we denote $KP(u)$ the \emph{prefix-free Kolmogorov complexity}, see \cite{li-vitanyi,suv} for details.
We use the standard properties of prefix-free Kolmogorov complexity:
\[
C(u) \le KP(u) + O(1) \le C(u) + O(\log C(u)).
\]
Thus, for all $u\in A$ we have $KP(u) \le m + O(k+\log m)$.
In what follows, we prove that, with probability close to $1$,
\[
KP(u)\ge  m - O(\Delta_1 + k+\log m),
\]
which implies the required bound for the plain Kolmogorov complexity $C(u)$.

We will need the \emph{Noiseless Coding Theorem} saying that for any distribution $\cal U$
\[
H({\cal U}) = \sum\limits_{u\in A} \pr[{\cal U} = u] \log \frac1{\pr[{\cal X} = u] } \le   \sum\limits_{u\in A} \pr[{\cal X} = u] \cdot KP(u), 
\]
see, e.g.,  \cite[Theorem 2.10]{grunwald}.
Finally, we denote by $h$ Shannon's entropy of $\cal U$, 
\[
h:= \sum\limits_{u\in A} \pr[{\cal U} = u] \log \frac1{\pr[{\cal X} = u] },
\]
and by $q$ the probability of the ``undesirable'' event,
\[
q:=\pr_{u\leftarrow \cal U}[KP(u)< m-\Delta_2].
\]
Now we are ready to estimate  the value of $q$.
\[
\begin{array}{rcl}
m - \Delta_1  &\le & H({\cal U}) \ \text{/* assumption of the lemma */}
\\
&=& \sum\limits_{u\in A} \pr[{\cal U} = u] \log \frac1{\pr[{\cal X} = u] }\ \text{/* definition of entropy */}
 \\
&\le&  \sum\limits_{u\in A} \pr[{\cal U} = u] \cdot KP(u) \ \text{/* Noiseless Coding Theorem */}
\\
& = &
\sum\limits_{u\in A \ :\ KP(u)< m-\Delta_2 } \pr[{\cal X} = u] \cdot KP(u)  
+ \sum\limits_{u\in A \ :\ KP(u)\ge m-\Delta_2 } \pr[{\cal X} = u] \cdot KP(u) 
\\
&\le &
\sum\limits_{u\in A \ :\ KP(u)< m-\Delta_2 } \pr[{\cal X} = u] \cdot (m-\Delta_2) +{} \\
&&
\hspace{10em} + \sum\limits_{u\in A \ :\ KP(u)\ge m-\Delta_2 } \pr[{\cal X} = u] \cdot (m + O(k)  +O (\log m))
\\
&\le &
q  \cdot (m-\Delta_2) 
+ (1-q) \cdot (m + O(k) +O (\log m)). \ \text{/* the definition of $q$ */}
\end{array}
\]
It follows that
\[
q\le \frac{ \Delta_1  + O(k) + O(\log m)}{ \Delta_2 } ,
\]
and it remains to choose $\Delta_2$ so that the last fraction is less than $1/10$.
\end{proof}

\end{document}